\documentclass{amsart}
\usepackage{amsmath}
\usepackage{amssymb,amsthm}
\numberwithin{equation}{section}
\newtheorem{thm}{Theorem}[section]

\newtheorem{lem}[thm]{Lemma}

\theoremstyle{definition}
\newtheorem{dfn}[thm]{Definition}
\theoremstyle{remark}
\newtheorem{rem}[thm]{Remark}

\numberwithin{equation}{section}

\begin{document}
\title[Born-Oppenheimer approximation for an atom]{Born-Oppenheimer approximation for an atom in constant magnetic fields}
\author[Sohei Ashida]{Sohei Ashida}
\address{Department of Mathematics\\
Graduate School of Science\\
Kyoto University\\
Kyoto\\
606 Japan}
\email{ashida@math.kyoto-u.ac.jp}

\begin{abstract}
We obtain a reduction scheme for the study of the quantum evolution of an atom in constant magnetic fields using the method developed by Martinez, Nenciu and Sordoni based on the construction of almost invariant subspace. In Martinez-Sordoni \cite{MaSo2} such a case is also studied but their reduced Hamiltonian includes the vector potential terms. In this paper, using the center of mass coordinates and constructing the almost invariant subspace different from theirs, we obtain the reduced Hamiltonian which does not include the vector potential terms. Using the reduced evolution we also obtain the asymptotic expantion of the evolution for a specific localized initial data, which verifies the straight motion of an atom in constatnt magnetic fields.
\end{abstract}

\maketitle
\section{Introduction}\label{firstsec}
In this paper we study the evolution of the system of one nucleus and some electrons which is neutral, i.e., the total electric charge is zero, in constant magnetic fields and external electric potentials. We construct the reduction scheme of the evolution of the system  to the evolution of center of mass of the particles with small error terms which is regarded as Born-Oppenheimer approximation. It is shown that the reduced evolution of the center of mass does not depend on the magnetic field, which implies the straight motion of the center of mass when the external potential does not exists.

The Hamiltonian of an atom with a nucleus and $N$ electorns moving in constant magnetic fields has the form
\begin{multline}\label{myeq1.1}
\hat P=\frac{1}{2m}(D_{x_1}-e_1A(x_1))^2+\sum_{i=2}^{N+1}\frac{1}{2m_e}(D_{x_i}-eA(x_i))^2\\
+\sum_{i<j}V_{ij}(x_i-x_j)+\sum_{i=1}^{N+1}V_i(x_i).
\end{multline}
Here $x_1 \in \mathbb{R}^3$ (resp., $m$) denotes the position (resp., the mass) of the nucleus, $x_j,j\geq2$ (resp., $m_e$) denote the position (resp., the mass)  of electrons and $V_{ij}$ (resp., $V_i$) are interaction (resp., external) potentials. $e_1$ (resp., $e$) denotes the charge of the nucleus (resp., electrons) and $A$ denotes the vector potential. 

There is a lot of literature on the quantum many body problem. Especially the large time asymptotics of the solution of the Schr\"odinger equation has been studied intensively and the complete classification of the asymptotic behavior called asyptotic completeness was proved for a large variety of potentials (see Derezi\'nski-G\'erard \cite{DeGe}). We do not deal with such a problem here, but we study the phenomenon which happens when there are differences of the mass of the particles instead. Roughly speaking the Born-Oppenheimer approximation saiys the following. Since the electrons are lighter than the nuclei, they move rapidly and adjust their state adiabatically as the nuclei moves more slowly. To develop the mathematical theory of the Born-Oppenheimer approximation, the ratio of electronic and nuclear mass is denoted by $h^2$ and regarded as a small parameter. Before we consider the case with magnetic fields, it is instructive to recall the results for the case without magnetic fields. In this case the Hamiltonian of some nuclei and electrons is written as 
$$P(h)=-h^2\Delta_x - \Delta_{y}+V(x,y),$$
where we denote the coordinates of the electrons by $y$ and that of the nuclei by $x$. Our purpose is to study the asymptotics of the solution of the Schr\"odinger equation $ih\partial_t \varphi =P(h)\varphi$ as $h\to 0$. From the intuitive description of the Born-Oppenheimer approximation above if the electrons are in bound states for the fixed nuclei, i.e. the bound states for $P_e(x):=-\Delta_{y}+V(x,y)$ at the initial time, we expect the electrons remain in the bound states even after time passes. This suggests there is an almost invariant subspace close to the electronic bound states under the evolution $e^{-itP(h)}$.

The almost invariant subspaces are described by the projections (see Nenciu \cite{Ne1, Ne2}). If an orthogonal projection $\Pi$ satisfies $[P(h),\Pi]=\mathcal O(h^{\infty})$, then $e^{-itP}\Pi=\Pi e^{-itP}+\mathcal O(h^{\infty}\lvert t \rvert)$ holds which means $\mathrm{Ran}\Pi$ is the almost invariant subspace. We expect there exists such a projection $\Pi$ such that $\Pi-\Pi_0=\mathcal O(h)$ where $\Pi_0=\int^{\oplus}\Pi_0(x)dx$ and $\Pi_0(x)$ is the spectral projection onto an arbitrarily chosen part of the discrete spectrum of $P_e(x)$ separated from the other part of the spectrum. In Nenciu \cite{Ne1, Ne2} the projections onto almost invariant subspaces were constructed by recurrence formula and another construction by pseudodifferential calculus was introduced by Helffer-Sj\"ostrand \cite{HeS} and Sj\"ostrand \cite{Sj}. In the case of Born-Oppenheimer approximation the almost invariant subspace does not seem to exist according to the physical intuition saying that the adiabatic decouppling becomes weaker and weaker when the energy increases. However for any cutoff function $\chi$ a projection $\Pi$ which satisfies $[P(h),\Pi]\chi(P)=\mathcal O(h^{\infty})$ is constructed by Sordoni \cite{So}. Using the projection $\Pi$ the quantum evolution $e^{-itP/h}$ of the molecule is reduced to the evolution of the nuclei $e^{-itG/h}$ where $G$ is a $k\times k$ matrix of semiclassical pseudodifferential operators $H^2(\mathbb R^n_x)\to L^2(\mathbb R^n_x)$, of the nuclei-variables, $k$ being the rank of $\Pi_0$ (see Martinez-Sordoni \cite{MaSo1, MaSo2}). The symbol of G is written as $g=\xi^2I_k+\mu(x)+\sum_{j=1}^{\infty}g_jh^j$, where $I_k$ is the $k$-dimentional identity matrix and $\mu(x)$ is a matrix of $\Pi_0P_e(x)$ in a basis of $\mathrm{Ran}\Pi_0$.

Let $f$ be a generalized coherent state of nuclei and $\psi$ be a bound state of electrons. The initial state of the form $f(x)\psi(x,y)+\mathcal O(h)$ is used by Hagedorn (see Hagedorn \cite{Ha1, Ha2, Ha3} and Hagedorn-Joye \cite{HaJo}). In semiclassical limit the quantum evolution of generalized coherent states admits an asymptotic expantion each term of which is a generalized coherent state centered at the point reached by the classical flow (see Combescure-Robert \cite{CoRo} and Combescure \cite{Co}). In \cite{MaSo2} Martinez and Sordoni applied this expantion to $G$ in order to obtain the evolution like that of Hagedorn.

We now return to the case with constant magnetic fields. Its coordinates perpendicular to the magnetic field stay in a bounded region as in the classical mechanics. When there are $N$ particles, the center of mass of the particles does not move freely, so that we cannot devide the motion of the particles into the internal and external motion. However, if the total charge of the particles is zero, there is a subspace $\mathcal H_{bound}$ of $L^2(\mathbb R^{3N})$ such that internal coordinates of the particles in the state in $\mathcal H_{bound}$ stay in a bounded region and the particles travel to infinity across the magnetic field (see G\'erard-\L aba \cite{GeLa}). This would occur from the viewpoint of the calssical mechanics because the Lorentz force on the particles with the opposite charge has the opposite direction and the particles interact with the interaction potential. Thus we expect when there is an atom, the electrons around the nucleus offset the influence of the magnetic fields. Hence it seems to be natural that the vector potential terms do not appear in the reduced Hamiltonian $G$. The Born-Oppenheimer approximation with magnetic fields is also dealt with by Martinez-Sordoni \cite{MaSo2} but in their construction the term $h^2(D_x-eA(x))^2$ remaines in $G$. This is because in thier coordinates where independent variables are positions of the particles, the effects of magnetic fields on the nuclei and the electrons are treated separately and therefore even if we restrict the Hamiltonian $P$ on $\mathrm{Ran}\Pi_0$, the effect of magnetic fields does not disappear. Our purpose is to construct a reduction scheme from $\hat P$ in \eqref{myeq1.1} to $G$ without vector potential terms.

To obtain such $G$ we change the coordinates $(x_1,\dots,x_{N+1})$ to the new coordinates where independent variables are the center of mass $x$ and the relative position $y=(y_2,\dots,y_{N+1})$ of the electrons from the nucleus. Then the Hamiltonian $\hat P$ is transformed into a certain operator $\tilde P$. Next we transform $\tilde P$ by a unitary transformation $\mathcal V=\exp\left(-ieA(x)\sum_{i=2}^{N+1}y_i\right)$ as $P:=\mathcal V\tilde P\mathcal V^*$. The transformed Hamiltonian is written in the following form:
$$P=h^2D_x^2-4h^2e\sum_{i=2}^{N+1}A(y_i)D_x+\sum_{i=2}^{N+1} L_i ^2+h^2Q+V(x,y).$$
Here $L_i=D_{y_i}-eA(y_i)$, the potential $V(x,y)$ is devided as $V(x,y)=V_0(x,y)+V'(x,y)$, where $V_0$ is the zeroth order term with respect to $h$, and $Q$ is a certain operator on $L^2(\mathbb R^{3N}_y)$. We regard $P_e(x)=\sum_{i=2}^{N+1}L_i^2+V_0(x,y)$ as the electric Hamiltonian and denote the spectral projection onto the chosen part of the discrete spectrum of $P_e(x)$ by $\Pi_0(x)$.

Our main results, Theorems \ref{invariant0}-\ref{coherent} in Section \ref{secondsec} can be illustrated as follows. For any cutoff function $\chi\in C_0^{\infty}(\mathbb R)$ there exists an orthogonal projection $\Pi(h)$ onto the almost invariant subspace close to $\Pi_0$, that is $[P,\Pi]\chi(P)=\mathcal O(h^{\infty})$ and $\Pi-\Pi_0=\mathcal O(h)$. There also exists a reducing transformation
$$W:L^2(\mathbb R^{3(N+1)})\to (L^2(\mathbb R^3_x))^{\oplus k},$$
where $k=\mathrm{Rank \Pi_0}$, and reduced Hamiltonian $G$ which is $k\times k$ self-adjoint matrix of operators on $L^2(\mathbb R^3_x)$ without vector potential term such that the restriction $U$ of $W$ to $\mathrm{Ran}\Pi$ is a unitary operator satisfying
$$e^{-itP/h}\Pi\varphi_0=U^*e^{-itG/h}U\Pi\varphi_0+\mathcal O(\lvert t\rvert h^{\infty}\lVert \varphi_0\rVert),$$
for $\varphi_0\in\mathrm{Ran}\chi(P)$. This means that the motions of all the particles are reduced to the motion of the center of mass.  Moreover, if $k=1$ and the initial data is
$$\tilde\varphi_0=(\pi h)^{-3/4}\tilde\Pi\chi(\tilde P)(e^{ix\xi_0/h-(x-x_0)^2/2h}\tilde u_1(x)),$$
where $\tilde \Pi:=\mathcal V^*\Pi\mathcal V$ and $\tilde u_1\in \mathrm{Ran}(\mathcal V^*\Pi_0)$, we have the following expansion for the evolution $e^{-it\tilde P/h}$:
$$e^{-it\tilde P/h}\tilde\varphi_0=e^{i\delta_t/h}\sum_{\mu=0}^{3(N-1)}c_{\mu}(t;h)\phi_{\mu,t}\tilde v_{\mu}(x)+\mathcal O(h^{N/4}),$$
where $\delta_t$ and $c_{\mu}(t;h)$ are constants depending on $t$ and $h$, $\phi_{\mu,t}$ is the generalized coherent states centered at the point reached by the classical flow of the symbol $g$ of $G$, and $\tilde v_{\mu}\in C^{\infty}(\mathbb R^3;L^2(\mathbb R^{3N}_y))$.

The difficulty in our construction of $\Pi$ is that when we expand the symbol of the resolvent $(P-z)^{-1}$ in formal power series in $h$ as $\sum_{j=0}^{\infty}q_j(x,\xi;z)h^j$ where $q_j$ are operators on $L^2(\mathbb R^{3N}_y)$, the power of $y$ contained in $q_j(z)$ becomes of higher order as $j$ increases. This difficulty is overcome by the exponential decay of the eighenfunctions of $P_e$ and the fact that the line integral $\hat \pi_j:=\oint _{\Gamma}q_j(z)dz$ where $\Gamma$ is a certain curve can be written as the sum of the terms containing $\Pi_0$. To prove the boundedness of $\chi(P)\hat \Pi_j$ where $\tilde\Pi_j:=Op_h^w(\tilde\pi_j)$, we also use the similar but slightly different expantion of the symbol of $(P-i)^{-1}$.

The result of G\'erard-\L aba \cite{GeLa} mentioned above also indicates that if the total charge is zero and the initial data is in a certain subspace of $L^2(\mathbb R^N)$, the particles can travel across the magnetic field. This result holds without any approximation but in our Born-Oppenheimer framework the localization is clearer, and although G\'erard-\L aba \cite{GeLa} considered the case without external potentials in which the pseudomomentum of the center of mass commutes with the Hamiltonian, our result covers the case with smooth bounded external potentials.

The content of this paper is as follows. In sec. 2 we transform the Hamiltonian into the form without vector potential term of nuclear variable, introduce our assumptions and state our main results. In sec. 3 we prove the main results. In the Appendix we collect the results we need on the pseudodifferential operators with operator valued symbols.

\section{Some preliminaries and main results}\label{secondsec}
We suppose the magnetic field is parallel to the third axis, so that the vector potential is written as
$$A(x)=\begin{pmatrix}
0 &-b &0\\
b &0 &0\\
0 &0 &0
\end{pmatrix}x.$$
where $b>0$ is a constant.
Setting the mass of electrons to 1 we introduce new coordinates $(x,y_2,\dots,y_{N+1})=(x,y)\in \mathbb R^3\times\mathbb R^{3N}$ by setting
\begin{align*}
&x=\frac{1}{M}\left (mx_1+\sum_{i=2}^{N+1}x_i\right ),\ M=m+N\\
&y_i=x_i-x_1,\ 2\leq i\leq N+1.
\end{align*}
Here $M$ is the total mass, $x$ is the position of center of mass and $y_i$ is the position of the electrons relative to the nucleus. With this choice of coordinates, we have the Hamiltonian
$$\tilde P=h^2D_x^2-2h^2e\sum_{i=2}^{N+1}A(y_i)D_x+\sum_{i=2}^{N+1}\tilde L_i(x) ^2+h^2\tilde Q+V(x,y)$$
where $h^2=\frac{1}{M},$ $\tilde L_i(x)=D_{y_i}-eA(y_i+x)$, and writing $f=\sum_{i=2}^{N+1}y_i$, $V(x,y)$ and $\tilde Q=\tilde Q_1+\tilde Q_2$ are as follows.
\begin{align*}
&V(x,y)=\sum_{i=2}^{N+1}V_{1i}(y_i)+\sum_{2\leq i<j}V_{ij}(y_i-y_j)\\
&\qquad\qquad+V_1\left(x-h^2f\right)+\sum_{i=2}^{N+1}V_i\left(x+y_i-h^2f \right),\\
&\tilde Q_1=\frac{1}{1-Nh^2}\left(\sum_{i=2}^{N+1}\tilde L_i(x)\right)^2,\\
&\tilde Q_2=\frac{1}{1-Nh^2}\left[2\sum_{i=2}^{N+1}\tilde L_i(x)\left(\sum_{i=2}^{N+1}eA\left(y_i+h^2f\right)\right)+\left(\sum_{i=2}^{N+1}eA\left(y_i+h^2f\right)\right)^2\right]\\
&\qquad+2\sum_{i=2}^{N+1}e\tilde L_i(x)A(f)+h^2e^2A(f)^2.
\end{align*}

We apply the unitary transformation $\mathcal V:=\exp\left(-ieA(x)\sum_{i=2}^{N+1}y_i\right)$ and its inverse to $\tilde P$ and obtain
\begin{align*}
P&=\mathcal V\tilde P\mathcal V^*\\
&=h^2D_x^2-4h^2e\sum_{i=2}^{N+1}A(y_i)D_x+\sum_{i=2}^{N+1} L_i ^2+h^2Q+V(x,y).
\end{align*}
Here, $L_i=D_{y_i}-eA(y_i)$, and $Q=Q_1+Q_2$, where
\begin{align*}
&Q_1=\frac{1}{1-Nh^2}\left(\sum_{i=2}^{N+1}L_i\right)^2,\\
&Q_2=\frac{1}{1-Nh^2}\left[2\sum_{i=2}^{N+1}L_i\left(\sum_{i=2}^{N+1}eA\left(y_i+h^2f\right)\right)+\left(\sum_{i=2}^{N+1}eA\left(y_i+h^2f\right)\right)^2\right]\\
&\qquad+2\sum_{i=2}^{N+1}eL_iA(f)+2h^2e^2A(f)^2.
\end{align*}
We write 
$$V_0(x,y)=\sum_{i=2}^{N+1}V_{1i}(y_i)+\sum_{2\leq i<j}V_{ij}(y_i-y_j)+V_1(x)+\sum_{i=2}^{N+1}V_i(x+y_i).$$
Note that $V_0(x,y)$ is the zeroth order terms with respect to $h$ of formal taylor expantion of $V(x,y)$.
We regard $P_e(x)=\sum_{i=2}^{N+1}L_i^2+V_0(x,y)$ as the electric Hamiltonian.
\begin{rem}
If $\varphi$ is a solution of $ih\partial_t\varphi=P\varphi$, then $\mathcal V^*\varphi$ is a solution of $ih\partial_t\varphi=\tilde P\varphi$. Since $\mathcal V^*$ is multiplication of the complex number with modulus 1, $\mathcal V^*\varphi$ has the same modulus as $\varphi$.
\end{rem}

We suppose the potentials $V_{ij}$ and $V_i$ satisfy the following assumptions.
\begin{enumerate}
\item[(H1)]
\begin{enumerate}
\item[(i)]$V_{ij}$ are real valued function $\Delta$-bounded with relative bound smaller than $1$.
\item[(ii)]$V_i\in C^{\infty}$ are real valued function and for any $\alpha\in\mathbb N^3$ there exist a constant $C_{\alpha}$ such that 
$$\lvert \partial^{\alpha}V_i(r)\lvert \leq C_{\alpha}.$$
\end{enumerate}
\end{enumerate}
By (H1) $\hat P$ is well defined as a self-adjoint operator with the domain $\{ u\in L^2(\mathbb R^{3(N+1)}):((D_{x_1}-e_1A(x_1))^2+\sum_{i=2}^{N+1}(D_{x_i}-eA(x_i))^2)u\in L^2(\mathbb R^{3(N+1)})\}$, so that $P$ is also self-adjoint with some domain. $P_e(x)$ can also be regarded as a self-adjoint operator with the domain $\{ u\in L^2(\mathbb R^{3N}): \sum_{i=2}^{N+1}L_i^2u\in L^2(\mathbb R^{3N})\}$. We also suppose that
\begin{enumerate}
\item[(H2)]The spectrum $\sigma (P_e(x))$ is the union of two disjoint components $\sigma_j(x),\ j=0,1$, such that $\sigma _0(x)$ is a part of discrete spectrum of $P_e(x)$ with the corresponding subspace of $L^2(\mathbb R^{3N})$ being finite dimentional and there exists a number $d>0$ such that
$$ \inf_{x\in \mathbb R^3}\mathrm{dist} (\sigma_0(x),\sigma_1(x))\geq d.$$
\end{enumerate}

\begin{rem}\label{loop}
By (H2) it is easy to see that there exist a continuous family of loops $\gamma(x)$ which encloses $\sigma_0(x)$ for each $x$ and $\inf_{x\in \mathbb R^3}\mathrm{dist}(\gamma(x),\sigma(P_e(x)))>d'$ for some $d'>0$.
\end{rem}

We denote by $\Pi_0(x)$ the spectral projection of $P_e(x)$ corresponding to $\sigma_0$. We also suppose the following assumption.
\begin{enumerate}
\item[(H3)]$\mathrm{Ran}\Pi_0(x)$ is spaned by a orthonormal basis $(u_1(x,y),\dots,u_k(x,y))\in C^{\infty}(\mathbb R^3;L^2_y)$ such that
$$\int_{\mathbb R^{3N}}\lvert u_i(x,y)\rvert ^2e^{2\alpha \lvert y\rvert}dy<C,$$
where constants $C>0$ and $\alpha>0$ do not depend on $x$.
\end{enumerate}

\begin{rem}
If $V_i(r)\equiv 0,\ 1\leq i\leq N+1$, then $\sigma _1(x)$ and $\Pi_0(x)$ don't depend on $x$, so that the assumptions (H2) and (H3) are obviously satisfied. If the fluctuation of $V_i$ is small enough, by the upper semicontinuity of the spectrum the eigenvalues in $\sigma_0$ do not cross each other, so that the corresponding eigenfunctions are smooth orthonormal basis. Moreover if $V_i$ are periodic, the eigenfunctions decay uniformly with respect to $x$ since $\Omega_{\epsilon}$ in the proof of Agmon \cite[Theorem 4.1]{Ag} can be chosen uniformly in $x$.
\end{rem}

Our main results are concerned with the almost invariant subspace which is close to electronic eigenspace. In the following theorem, $a=\mathcal O(h^{\infty})$ means $a=\mathcal O(h^N)$ for any $N\in\mathbb N$.
\begin{thm}\label{invariant0}
Assume (H1)-(H3) hold true. Then for any $\Phi\in C_0^{\infty}(\mathbb R)$ there exists a orthogonal projection $\Pi(h)$ on $L^2(\mathbb R^{3(N+1)})$ such that
$$\lVert \Pi-\Pi_0\rVert_{L^2(\mathbb R^{3(N+1)})}=\mathcal O(h),$$
and, for any $\chi\in C_0^{\infty}(\mathbb R)$ such that $\chi\phi=\chi$, we have 
$$\lVert \chi(P)[\Pi,P]\rVert_{L^2(\mathbb R^{3(N+1)})}+\lVert [\Pi,P]\chi(P)\rVert_{L^2(\mathbb R^{3(N+1)})}=\mathcal O(h^{\infty}).$$
\end{thm}

\begin{thm}\label{evolutionas}
If $\varphi$ is the solution of $ih\partial_t\varphi=P\varphi$ with initial data $\varphi_0$ satisfying $\chi(P)\varphi_0=\varphi_0$ for some $\chi\in C_0^{\infty}(\mathbb R)$ such that $\chi\phi=\chi$, then 
\begin{equation}\label{myeq2.11}
\varphi =e^{-itP_1/h}\Pi\varphi_0+e^{-itP_2/h}(1-\Pi)\varphi_0+\mathcal O(\lvert t\rvert h^{\infty}\lVert \varphi_0\rVert),
\end{equation}
where $P_1:=\Pi P\Pi$ and $P_2=(1-\Pi)P(1-\Pi)$ is self-adjoint on a domain containing $D(P)$.
\end{thm}

We can reduce the evolution $e^{-itP_1}$ to an evolution on the $L^2$ space of only nuclear variables. To state the next result we use h-admissible operators. For the definition of h-admissible operators see the the Appendix.

\begin{thm}\label{reduction}
There exists a h-admissible operator
$$W:L^2(\mathbb R^{3(N+1)})\to (L^2(\mathbb R^3_x))^{\oplus k}$$
with operator valued symbol and $k\times k$ self-adjoint matrix $G$ of h-admissible operators on $L^2(\mathbb R^3_x)$ such that the restriction $U$ of $W$ to $\mathrm{Ran}\Pi$:
$$U:\mathrm{Ran}\Pi \to (L^2(\mathbb R^3_x))^{\oplus k}$$
is a unitary operator which satisfies 
$$UP_1\Pi=GU\Pi,$$
so that $e^{-itP_1/h}\Pi=U^*e^{-itG/h}U\Pi$. The symbol $g(x,\xi)$ of $G$ has the following form:
$$g(x,\xi)=\xi^2I_k+ \mu(x) +\sum_{i\geq 1} h^jg_j(x,\xi),$$
where $\mu (x)$ is the matrix of $\Pi_0(x)P_e(x)$ in $(u_1(x),\dots,u_k(x))$ and $g(x,\xi)$ does not contain vector potential terms.
\end{thm}

\begin{rem}
When more than one nucleus exists, these results with the operator $G$ without vector potential term does not seem to hold because each electron is possessed by more than one nucleus at the same time and the electrons around each nucleus does not necessarily offset the influence of magnetic fields.
\end{rem}

The next result is concerned with more explicit expression of the solution of the Schr\"odinger equation for a special initial data. Let us write $\tilde u_j(x):=\mathcal V^*u_j(x)$. Then $(\tilde u_1(x),\dots,\tilde u_k(x))$ are the orthnormal basis of $\mathrm{Ran}\tilde\Pi_0(x)$ where
$$\tilde \Pi_0(x):=\mathcal V^*\Pi_0(x)\mathcal V.$$
which is the spectral projection of $\tilde P_e(x)=\sum_{i=2}^{N+1}\tilde L_i(x)^2+V_0(x,y)$ corresponding to $\sigma_0(x)$.

Let $\alpha_t=(x_t,\xi_t)$ be the solution of
\begin{equation}\label{myeq2.12}
x_t=\frac{\partial g}{\partial \xi}(x_t,\xi_t),\ \xi_t=\frac{\partial g}{\partial x}(x_t,\xi_t)
\end{equation}
starting from initial data $\alpha_0=(x_0,\xi_0)$. Let $(\eta_n,\zeta_n),\ n=1,2,3$ be the independent solutions of 
\begin{equation}
\begin{pmatrix}
\dot \eta\\
\dot \zeta
\end{pmatrix}
=JM_t
\begin{pmatrix}
\eta\\
\zeta
\end{pmatrix},
\end{equation}
with initial data
$$(\eta_n)_j\vert_{t=0}=\delta_{jn},\ (\zeta_n)_j\vert_{t=0}=i\delta_{jn},$$
where $(\eta_n)_j$ is the jth component of $\eta_n$,
$$J=
\begin{pmatrix} 0&\mathbb I\\ -\mathbb I&0\end{pmatrix},$$
$\mathbb I$ being the unit matrix and $M_t$ is the Hessian of $g$ at $\alpha_t$:
$$(M_t)_{i,j}=\left(\frac{\partial^2g}{\partial\alpha^2}\right)_{i,j}\bigg \vert_{\alpha=\alpha_t}.$$
By Theorem \ref{evolutionas} and Theorem \ref{reduction} as in Martinez-Sordoni \cite[Theorem 11.3]{MaSo2} (see also Combescure-Robert \cite{CoRo}), we have the following theorem.

\begin{thm}\label{coherent}
Let $k=1$ and $\tilde\varphi_0\in L^2(\mathbb R^{3(N+1)})$ be as follows
$$\tilde\varphi_0=(\pi h)^{-3/4}\tilde\Pi\chi(\tilde P)(e^{ix\xi_0/h-(x-x_0)^2/2h}\tilde u_1(x)),$$
where $\tilde \Pi=\mathcal V^*\Pi\mathcal V$, and $\chi=1$ near $\xi_0^2+ \mu(x_0)$.
Then there exists $C>0$ such that for any integer $J\geq1$ one has
\begin{equation}\label{myeq2.13}
e^{-it\tilde P/h}\tilde\varphi_0=e^{i\delta_t/h}\sum_{\mu=0}^{3(J-1)}c_{\mu}(t;h)\phi_{\mu,t}\tilde v_{\mu}(x)+\mathcal O(h^{J/4}),
\end{equation}
where
$$c_{\mu}(t;h)=\sum_{j=0}^{J_{\mu}}h^{j/2}c_{\mu,j}(t),$$
$c_{\mu,j}$ are polynomials with respect to $\partial^{\gamma}g(x_t,\xi_t)$ and $\mathrm{Re}\eta_n,\mathrm{Im}\eta_n,\mathrm{Re}\zeta_n,\mathrm{Im}\zeta_n\ 1\leq n\leq 3$,$\delta_t:=\int_0^t(\dot x_s\xi_s-g(x_s,\xi_s))ds+(x_0\xi_0-x_t\xi_t)/2$ and $\tilde v_{\mu}\in C^{\infty}(\mathbb R^3;L^2_y(\mathbb R^{3N}))$. The estimate is uniform with respect to $(t,h)$ such that $h>0$ is small enough and $t<C^{-1}\ln \frac{1}{h}$. Moreover, the lowest order term with respect to $h$ is $(\pi h)^{-3/4}(e^{ix\xi_t/h-(x-x_t)^2/2h}\tilde u_1(x))$.
\end{thm}

\section{Proof of the main results}\label{thirdsec}
The outline of the construction of $\Pi$ can be illustrated as follows. First we obtain the expression of the symbol of the resolvent $(P-z)^{-1}$ of $P$ as formal power siries $\sum_{j=0}^{\infty}q_j(x,\xi;z)h^j$ with respect to $h$. Secondly we integrate each $q_j$ with respect to $z$ along the loop $\Gamma(x,\xi)$ enclosing the set $\{z:z-\xi^2\in\sigma_0(x)\}$ and denote it by $\hat\pi_j(x,\xi)$. Thirdly we quantize $\hat\pi_j$ and multiply the cutoff function $\Phi(P)$ to make each term bounded. Fourthly we make a resummation and symmetrize the operator to obtain a operator $\hat\Pi_{\phi}$. Finally integrating the resolvent of $\hat\Pi_{\phi}$ along the loop enclosing the point $z=1$, we obtain $\Pi$.

We denote by $p(x,\xi;h)=\xi^2+\sum_{i=2}^{N+1} L_i ^2-4h^2e\sum_{i=2}^{N+1}A(y_i)\xi+h^2Q+V(x,y)$ the symbol of $P$ and by $p_0(x,\xi):=\xi^2+\sum_{i=2}^{N+1} L_i ^2+V_0(x,y)$ its principal symbol. Let $\gamma (x)$ be the continuous loop as in Remark \ref{loop}. Let us consider $\Omega:=\lbrace (x,\xi,z)\in\mathbb R^6\times \mathbb C\ ;z-\xi^2\in \gamma(x)\rbrace$. For $(x,\xi,z)\in \Omega$, $p_0(x,\xi)-z$ is invertible and $q_0:=(p_0-z)^{-1}$ is smooth and bounded. We define the symbol
$$r(x,\xi;h,z)=\sum_{j\geq 1}r^j(x,\xi;z)h^j,$$
as in Sordoni \cite{So} (see also Nenciu-Sordoni \cite{NeSo}) by
$$(p(x,\xi;h)-z)\#q_0(x,\xi;z)=1-r(x,\xi;h,z),$$
where 
$$a(x,\xi)\# b(x,\xi)=\sum_{\alpha,\beta}\frac{h^{\lvert \alpha+\beta\rvert}(-1)^{\lvert \alpha \rvert}}{(2i)^{\lvert \alpha+\beta\rvert}\alpha !\beta !}(\partial_x^{\alpha}\partial_\xi^{\beta}a(x,\xi))(\partial_\xi^{\alpha}\partial_x^{\beta}b(x,\xi)),$$
at a formal series level.
We define 
$$q(x,\xi;h,z)=\sum_{j\geq 1}q^j(x,\xi;z)h^j,$$
by
$$q:=q_0+q_0\#\sum_{j\geq1}r^{\#j}.$$
Here $r^{\#j}=r\#\dotsm\# r$ where the number of $\#$ is $j$.
One can check that, for $j\geq 1$, $q_j$ is given by the sum of the terms of the following form
\begin{equation}\label{myeq4.00}
Cq_0\prod_{i=1}^{m}((\partial_x^{\alpha_i}\partial_{\xi}^{\beta_i}s_i)q_0),
\end{equation}
where $C$ is a constant, $1\leq m\leq 2j$, $\lvert \sum_{i=1}^{m}\alpha_i\rvert=\lvert \sum_{i=1}^{m}\beta_i\rvert\leq j$ and $s_i$ is  one of the following

\begin{equation}\label{myeq4.0}
\begin{split}
&\xi^{\alpha},\ \lvert\alpha\rvert=1,\quad \sum_{i=2}^{N+1}A(y_i)\xi,\quad L_iA(y_i),\quad \left(\sum_{i=2}^{N+1}L_i\right)^2,\quad \left(\sum_{i=2}^{N+1}eA(y_i)\right)^2,\\
&\partial_x^{\gamma}V_1(x)\left(\sum_{i=2}^{N+1}y_i\right)^{\gamma}+\sum_{i=2}^{N+1}\partial_x^{\gamma}V_i(x+y_i)\left(\sum_{i=2}^{N+1}y_i\right)^{\gamma},\ \gamma\in\mathbb N^3.
\end{split}
\end{equation}

Since the operators do not commutate with each other, the order of the product is specified as $\prod_{i=1}^{m}((\partial_x^{\alpha_i}\partial_{\xi}^{\beta_i}s_i)q_0)=(\partial_x^{\alpha_1}\partial_{\xi}^{\beta_1}s_1)q_0\dotsm (\partial_x^{\alpha_m}\partial_{\xi}^{\beta_m}s_m)q_0$. 
Let us set
$$\Gamma (x,\xi):=\lbrace z\in\mathbb C;z-\xi^2\in \gamma(x)\rbrace,$$
and let us define
\begin{equation}\label{pi}
\hat\pi_j(x,\xi)=\frac{i}{2\pi}\oint_{\Gamma (x,\xi)}q_j(x,\xi;z)dz.
\end{equation}
Then by the change of the variable $z\to z-\xi^2$, we see $\hat\pi_0(x,\xi)=\Pi_0(x)$.

If a potential $V$ in $\mathbb R^3$ is $\Delta$-bounded with relative bound $\epsilon$, then $V(y_i)$ and $V(y_i-y_j)$ are relatively bounded with respect to $\sum_{i=2}^{N+1}L_i^2$ with relative bound $\epsilon$. Thus we see easily the following lemma.

\begin{lem}\label{bounded}
Let $R_e(x;z)=(P_e(x)-z)^{-1}$ be the resolvent of $P_e$. Then
\begin{itemize}
\item[(i)] $(L_i)_kR_e(x;z)$ and $L_iL_jR_e(x;z)$ are bounded as operators on $L^2(\mathbb R^{3N}_y)$ uniformly with respect to $x$ where $(L_i)_k,\ k=1,2,3$ denotes the $k$th component of $L_i$.
\item[(ii)] $(L_i)_kq_0,L_iL_jq_0,\xi^2 q_0$ and $\lvert\xi\rvert (L_i)_kq_0$ are bounded as operators on $L^2(\mathbb R^{3N}_y)$ uniformly with respect to $x,\xi$.
\end{itemize}
\end{lem}
To quantize $\hat\pi_j$ we need to prove that they satisfy the condition of the symbols (see the Appendix).

\begin{lem}\label{bound1}
For any $j\in \mathbb N$ and $r,s\in \mathbb R$, $\hat\pi_j(x,\xi)\in S^j(\mathbb R^3\times \mathbb R^3;\mathcal L(L^{2,r}_y,L^{2,r+s}_y))$ where $L^{2,s}_y:=\{ u: (1+\lvert y\rvert^2)^{s/2}u\in L^2_y\}$.
\end{lem}

\begin{proof}
By \eqref{myeq4.00} $q_j$ can be written as the sum of the terms of the following form
\begin{equation}\label{myeq4.01}
Cq_0(\prod_{i=1}^{m}(y^{\alpha_i}a_iq_0))\xi^{\beta},
\end{equation}
where $C$ is a constant, $\alpha_i\in \mathbb N^{3N}$, $\beta\in \mathbb N^3$ and $a_iR_e(x;i)\in \mathcal L(L^2(\mathbb R_y^{3N}))$ uniformly with respect to $x$. It is easy to see $\lvert\beta\rvert\leq j$ from the construction of $q_j$.
By the change of the variable $z\to z-\xi^2$, $\oint_{\Gamma(x,\xi)}q_0\prod_{i=1}^{m}(y^{\alpha_i}a_iq_0)dz$ becomes
$$\oint_{\gamma(x)}R_e(x;z)\prod_{i=1}^{m}(y^{\alpha_i}a_iR_e(x;z))dz.$$
Let $\gamma_1(x)$ be a slightly larger loop enclosing $\gamma(x)$ and $z_1$ is a point on $\gamma_1(x)$. By the resolvent equation we have
\begin{align*}
&\oint_{\gamma(x)}\oint_{\gamma_1(x)}R_e(x;z_1)R_e(x;z)\prod_{i=1}^{m}(y^{\alpha_i}a_iR_e(x;z))dzdz_1\\
&\quad=\oint_{\gamma(x)}\oint_{\gamma_1(x)}\frac{1}{z-z_1}(R_e(x;z)-R_e(x;z_1))\prod_{i=1}^{m}(y^{\alpha_i}a_iR_e(x;z))dzdz_1.
\end{align*}
Since $\gamma_1(x)$ lies outside $\gamma(x)$, we have
\begin{align*}
&\oint_{\gamma(x)}\oint_{\gamma_1(x)}\frac{1}{z-z_1}R_e(x;z)\prod_{i=1}^{m}(y^{\alpha_i}a_iR_e(x;z))dzdz_1\\
&\quad=-2\pi i\oint_{\gamma(x)}R_e(x;z)\prod_{i=1}^{m}(y^{\alpha_i}a_iR_e(x;z))dz.
\end{align*}
Thus we have
\begin{align*}
&\oint_{\gamma(x)}\oint_{\gamma_1(x)}R_e(x;z_1)R_e(x;z)\prod_{i=1}^{m}(y^{\alpha_i}a_iR_e(x;z))dzdz_1\\
&\quad=-2\pi i\oint_{\gamma(x)}R_e(x;z)\prod_{i=1}^{m}(y^{\alpha_i}a_iR_e(x;z))dz.\\
&\qquad -\oint_{\gamma(x)}\oint_{\gamma_1(x)}\frac{1}{z-z_1}R_e(x;z_1)\prod_{i=1}^{m}(y^{\alpha_i}a_iR_e(x;z))dzdz_1.
\end{align*}
Transposing the terms and using $\Pi_0(x)=\frac{i}{2\pi}\oint_{\gamma(x)}(P_e-z)^{-1}dz$, we have
\begin{align*}
&\oint_{\gamma(x)}R_e(x;z)\prod_{i=1}^{m}(y^{\alpha_i}a_iR_e(x;z))dz\\
&\quad=\oint_{\gamma(x)}\Pi_0(x)R_e(x;z)\prod_{i=1}^{m}(y^{\alpha_i}a_iR_e(x;z))dz\\
&\qquad+\frac{i}{2\pi}\oint_{\gamma(x)}\oint_{\gamma_1(x)}\frac{1}{z-z_1}R_e(x;z_1)\prod_{i=1}^{m}(y^{\alpha_i}a_iR_e(x;z))dzdz_1.
\end{align*}
Repeating the same procedure for the rest of $R_e(x;z)$ we have
\begin{equation}\label{myeq4.6}
\begin{split}
&\oint_{\gamma(x)}R_e(x;z)\prod_{i=1}^{m}(y^{\alpha_i}a_iR_e(x;z))dz\\
&\quad=\sum_{\ell=0}^m\left(\frac{i}{2\pi}\right)^{\ell}\oint_{\gamma(x)}\oint_{\gamma_1(x)}\dotsm\oint_{\gamma_\ell(x)}\left(\prod_{n=1}^{\ell}\frac{1}{z-z_n}R_e(x;z_n)(y^{\alpha_n}a_n)\right)\\
&\qquad\cdot\Pi_0(x)\left(\prod_{i=\ell+1}^{m}R_e(x;z)(y^{\alpha_i}a_i)\right)R_e(x;z)dzdz_1\dotsm dz_{\ell}\\
&\qquad+\left(\frac{i}{2\pi}\right)^{m+1}\oint_{\gamma(x)}\oint_{\gamma_1(x)}\dotsm\oint_{\gamma_{m+1}(x)}\prod_{n=1}^{m+1}\frac{1}{z-z_n}\\
&\qquad\quad \cdot\prod_{i=1}^m(R_e(x;z_i)(y^{\alpha_i}a_i))R_e(x;z_{m+1})dzdz_1\dotsm dz_{m+1},
\end{split}
\end{equation}
where $\gamma_i(x)$ are the same loop as $\gamma_1(x)$. Since $\gamma_i(x)$ lies outside $\gamma(x)$ the integration with respect to $z$ in the last term of the right hand side vanishes.

By the assumption (H3)  $y^{\beta}\Pi_0(x)$ is a bounded operator on $L^2(\mathbb R^{3N}_y)$ uniformly for $x$. By $$[(y_i)_k,R_e(x;z)]=R_e(x;z)[P_e,(y_i)_k]R_e(x;z)=-2iR_e(x;z)(L_i)_kR_e(x;z),$$
the remaining term of \eqref{myeq4.6} are the sum of the terms of the following form;
\begin{multline}\label{myeq4.8}
C\oint_{\gamma(x)}\oint_{\gamma_1(x)}\dotsm\oint_{\gamma_{\ell}(x)}\left(\prod_{n=1}^{\ell}\frac{1}{z-z_n}\left(\prod_{j=1}^{j_n}R_e(x;z_n)b_{n,j}\right)\right)\\
\quad\cdot y^{\beta}\Pi_0(x)y^{\delta}\left(\prod_{i=1}^{r}R_e(x;z)c_i\right)R_e(x;z)dzdz_1\dotsm dz_{\ell},
\end{multline}
where $C$ is a constant, $y^{\beta}\Pi_0(x)y^{\delta}$ is bounded uniformly in $x$ and $R_e(x;i)b_{n,j}$, $\ R_e(x;i)c_j$ are bounded uniformly in $x$. Since 
$$(P_e(x)-i)R_e(x;z)=1+(z-i)R_e(x;z),$$
 is uniformly bounded for $x,\ z\in \gamma(x)$ and $(z-z_k)^{-1}$ is uniformly bounded in $z\in \gamma(x),\ z_k \in \gamma_k(x)$, \eqref{myeq4.8} is bounded uniformly in $x$. Thus by \eqref{myeq4.01} $\lVert\hat\pi_j(x,\xi)\rVert_{L^2_y}<C(1+\lvert\xi\rvert)^{j}$ uniformly for $x$, where $C>0$ is a positive constant. In the same way we can see for all $\alpha\in \mathbb N^{3N}$ and $\beta,\delta\in \mathbb N^3$ there exists a constant $C_{\alpha,\beta,\delta}$ such that $\lVert y^{\alpha}\partial_x^{\beta}\partial_{\xi}^{\delta}\hat\pi_j(x,\xi)\rVert_{L^2_y}<C_{\alpha,\beta,\delta}(1+\lvert\xi\rvert)^{j-\lvert\delta\rvert}$ . This completes the proof of the lemma.
\end{proof}

\begin{rem}
Since in Martinez-Sordoni \cite{MaSo2} they use the coordinates where independent variables are positions of the particles, the power of $y$ appears only in the electric Hamiltonian and is not included in $q_j$. On the other hand in our coordinate the power of $y$ is included in $4h^2e\sum_{i=2}^{N+1}A(y_i)D_x, Q$, and the taylor expantion of $V_i$. Thus in our case the power of $y$ in $q_j$ becomes of higher order as $j$ increases, so that the boundedness of $\hat\pi_j$ is not obvious.
\end{rem}

We define
$$\hat\Pi_j=Op_h^w(\hat\pi_j(x,\xi)).$$
The operator $\hat\Pi_j$ is not bounded but by a localization in energy we obtain a bounded operator.

\begin{lem}\label{bound2}
For any $\Phi\in C_0^{\infty}(\mathbb R)$, $\Phi(P)\hat\Pi_j$ is bounded in $L^2(\mathbb R^{3(N+1)})$.
\end{lem}
\begin{proof}
We can write
\begin{equation}\label{myeq4.80}
\Phi(P)\hat\Pi_j=\Phi(P)(P-i)^j(P-i)^{-j}\hat\Pi_j,
\end{equation}
where $\Phi(P)(P-i)^j$ is bounded. We set
$$p'(x,\xi):=p_0(x,\xi)+h^2Q_1.$$
We define
$$q'_0(x,\xi;h):=(p'(x,\xi)-i)^{-1},$$
and $r'(x,\xi;h)$ by
$$(p(x,\xi;h)-i)\#q'_0(x,\xi;h)=1-r'(x,\xi;h).$$
Since 
$$\partial_{x_k}q'_0(x,\xi;h)=-q'_0(x,\xi;h)\partial_{x_k}V_0q'_0(x,\xi;h),$$
and
$$\partial_{\xi_k}q'_0(x,\xi;h)=-q'_0(x,\xi;h)\xi_kq'_0(x,\xi;h),$$
by Lemma \ref{bounded} we see $r'(x,\xi;h)\in S^{-1}(\mathbb R^3\times\mathbb R^3;\mathcal L(L^{2,s}_y,L^{2,s-2}_y))$, for any $s\in \mathbb R$.
Since
$$(1-r')\#\left(\sum_{i=0}^{j-1}r'^{\#i}\right)=1-r'^{\#j},$$
we have
$$(P-i)Op_h^w\left(q_0'(x,\xi)\#\left(\sum_{i=0}^{j-1}r'^{\#i}\right)\right)=1-Op_h^w\left(r'^{\#j}\right).$$
Multiplying the both sides of the equation by $(P-i)^{-1}$ from left we obtain
$$(P-i)^{-1}=(P-i)^{-1}Op_h^w\left(r'^{\#j}\right)+Op_h^w\left(q_0'(x,\xi)\#\left(\sum_{i=0}^{j-1}r'^{\#i}\right)\right).$$
Hence we have
\begin{equation}\label{myeq4.81}
\begin{split}
(P-i)^{-j}\hat\Pi_j&=\sum_{k=0}^{j-1}(P-i)^{-j+k}Op_h^w\left(r'^{\#j}\right)\left(Op_h^w\left(q_0'(x,\xi)\#\left(\sum_{i=0}^{j-1}r'^{\#i}\right)\right)\right)^k\\
&\quad\cdot\hat\Pi_j+\left(Op_h^w\left(q_0'(x,\xi)\#\left(\sum_{i=0}^{j-1}r'^{\#i}\right)\right)\right)^j\hat\Pi_j.
\end{split}
\end{equation}
Since for any $r\in\mathbb R$
$$r'^{\#j}\#\left (q_0'(x,\xi)\#\left(\sum_{i=0}^{j-1}r'^{\#i}\right)\right)^{\#k}\in S^{-j-k}(\mathbb R^3\times \mathbb R^3;\mathcal L(L^{2,r}_y,L^{2,r-2j-2k}_y)),$$
and
$$\left(q_0'(x,\xi)\#\left(\sum_{i=0}^{j-1}r'^{\#i}\right)\right)^{\#j}\in S^{-j}(\mathbb R^3\times \mathbb R^3;\mathcal L(L^{2,r}_y,L^{2,r-2j}_y)),$$
by Lemma \ref{bound1} we have
$$r'^{\#j}\#\left (q_0'(x,\xi)\#\left(\sum_{i=0}^{j-1}r'^{\#i}\right)\right)^{\#k}\#\hat \pi_j\in S^{0}(\mathbb R^3\times \mathbb R^3;\mathcal L(L^{2}_y)),$$
$$\left(q_0'(x,\xi)\#\left(\sum_{i=0}^{j-1}r'^{\#i}\right)\right)^{\#j}\#\hat \pi_j\in S^{0}(\mathbb R^3\times \mathbb R^3;\mathcal L(L^{2}_y)).$$
Thus by Theorem \ref{L^2boundedness} it follows that
\begin{align*}
Op_h^w\left(r'^{\#j}\right)\left(Op_h^w\left(q_0'(x,\xi)\#\left(\sum_{i=0}^{j-1}r'^{\#i}\right)\right)\right)^k\hat\Pi_j&\in \mathcal L(L^2(\mathbb R^3_x;L^2(\mathbb R^{3N}_y)))\\
&=\mathcal L(L^2(\mathbb R^{3(N+1)})),
\end{align*}
\begin{equation}\label{myeq4.82}
\left(Op_h^w\left(q_0'(x,\xi)\#\left(\sum_{i=0}^{j-1}r'^{\#i}\right)\right)\right)^j\hat\Pi_j\in\mathcal L(L^2(\mathbb R^{3(N+1)})).
\end{equation}
By \eqref{myeq4.80}, \eqref{myeq4.81} and \eqref{myeq4.82}, we have $\Phi(P)\hat\Pi_j\in\mathcal L(L^2(\mathbb R^{3(N+1)}))$
\end{proof}

\begin{rem}
In Martinez-Sordoni \cite{MaSo2} they considered only localized initial data, so that unboundedness of the vector potential term is not important. On the other hand, in our case we do not restrict the initial data and therefore, because of the existence of the vector potential terms in $\hat P$ in \eqref{myeq1.1}, $D_x$ is not $P$-bounded. Thus, to prove $(P-i)^{-j}\hat\Pi_j$ is bounded, we need to expand the resolvent at the expense of the presence of the power of $y$. Since $\hat\pi_j$ includes $\Pi_0(x)$, the power of $y$ does not prevent $(P-i)^{-j}\hat\Pi_j$ from being bounded.
\end{rem}

To construct $\Pi$ we make a resummation. Let $\rho\in C_0^{\infty}(\mathbb R)$ be a function such that $0\leq\rho\leq1$, $\mathrm{supp}\ \rho\subset [-2,2]$ and $\rho=1$ on $[-1,1]$. It is easy to see that $P\hat\Pi_j\Phi(P)$ and $\Phi(P)\hat\Pi_jP$ are bounded in $L^2(\mathbb R^{3(N+1)})$ and there exists a decreasing sequence of positive numbers $(\epsilon_j)_{j\in \mathbb N}$ converging to zero such that for any $j\in\mathbb N$
$$\left(1-\rho\left(\frac{\epsilon_j}{h}\right)\right)(\lVert \hat\Pi_j\Phi(P)\rVert+\lVert \Phi(P)\hat\Pi_j\rVert+\lVert P\hat\Pi_j\Phi(P)\rVert+\lVert \Phi(P)\hat\Pi_jP\rVert)\leq h^{-1}.$$
Let us define
$$\hat\Pi\Phi(P):=\Pi_0\Phi(P)+\sum_{j\geq 1}\left(1-\rho\left(\frac{\epsilon_j}{h}\right)\right)\hat\Pi_j\Phi(P)h^j,$$
$$\Phi(P)\hat\Pi:=\Phi(P)\Pi_0+\sum_{j\geq 1}\left(1-\rho\left(\frac{\epsilon_j}{h}\right)\right)\Phi(P)\hat\Pi_jh^j.$$
Then it is easy to prove that for any $N\in \mathbb N$ there exists $C_N>0$ such that
\begin{equation}\label{myeq4.9}
\left\lVert\hat\Pi\Phi(P)-\sum_{j=0}^Nh^j\hat\Pi_j\Phi(P)\right\rVert+\left\lVert\Phi(P)\hat\Pi-\sum_{j=0}^Nh^j\Phi(P)\hat\Pi_j\right\rVert\leq C_Nh^{N+1}.
\end{equation}
and
\begin{equation}\label{myeq4.91}
\left\lVert P\hat\Pi\Phi(P)-\sum_{j=0}^Nh^jP\hat\Pi_j\Phi(P)\right\rVert+\left\lVert\Phi(P)\hat\Pi P-\sum_{j=0}^Nh^j\Phi(P)\hat\Pi_jP\right\rVert\leq C_Nh^{N+1},
\end{equation}
where we understand $\hat\Pi_0=\Pi_0$.
We define
$$\hat\Pi_{\Phi}:=\Phi(P)\hat\Pi+(1-\Phi(P))\hat\Pi\Phi(P)+(1-\Phi(P))\Pi_0(1-\Phi(P)).$$
Since $\sum_{j\geq 0}q_jh^j$ is formally the symbol of $(P-z)^{-1}$, by \eqref{myeq4.9} and \eqref{myeq4.91} we have the following lemma as in Sordoni \cite{So}.

\begin{lem}\label{invariant}
Let $\chi\in C_0^{\infty}(\mathbb R)$ satisfy $\chi\phi=\chi$. Then for all $N\in\mathbb N$
$$\lVert \chi(P)[\hat\Pi_{\Phi},P]\rVert+\lVert[\hat\Pi_{\Phi},P] \chi(P)\rVert=\mathcal O(h^N),$$
$$\lVert\chi(P)(\hat\Pi_{\Phi}^2-\hat\Pi_{\Phi})\rVert+\lVert(\hat\Pi_{\Phi}^2-\hat\Pi_{\Phi})\chi(P)\rVert=\mathcal O(h^N).$$
\end{lem}

Since $\hat\Pi_{\Phi}-\Pi_0=\mathcal O(h)$, for sufficiently small $h$ the spectrum of $\hat\Pi_{\Phi}$ is concentrated near $0$ and $1$, so that the set $\lbrace z\in \mathbb C;\lvert z-1\rvert=1/2\rbrace$ is in the resolvent set of $\hat\Pi_{\Phi}$ for sufficiently small $h$. We define
$$\Pi:=\oint_{\lvert z-1\rvert=1/2}(\hat\Pi_{\Phi}-z)^{-1}dz.$$
Then as in Sordoni \cite{So} we have
$$\Pi-\hat\Pi_{\Phi}=\frac{i}{2\pi}(\hat\Pi_{\Phi}^2-\hat\Pi_{\Phi})\oint_{\lvert z-1\rvert=1/2}(\hat\Pi_{\Phi}-z)^{-1}(2\hat\Pi_{\Phi}-1)(1-\hat\Pi_{\Phi}-z)^{-1}(1-z)^{-1}dz.$$
Theorem \ref{invariant0} follows from this formula and Lemma \ref{invariant}.

\begin{proof}[Proof of Theorem \ref{evolutionas}]
First, we prove that $\Pi P\Pi$ and $(1-\Pi) P(1-\Pi)$ are self-adjoint on a domain containing $D(P)$. Since it can be proved that $P\Pi(P-i)^{-1}$ is bounded as in the proof of Lemma \ref{bound2}, we have $\Pi D(P)\in D(P)$. Hence $\Pi P\Pi$ and $(1-\Pi) P(1-\Pi)$ are defined on $D(P)$. Since $\hat P$ is lower semibounded, $P$ is also lower semibounded. Thus $\Pi P\Pi$ and $(1-\Pi) P(1-\Pi)$ are lower semibounded so that there exist Friedrichs extentions of $\Pi P\Pi$ and $(1-\Pi) P(1-\Pi)$ with their domains containing $D(P)$.

Multiplying both side of $ih\partial_t\varphi=P\varphi$ by $\Pi$ we have
$$ih\partial_t\Pi\varphi=P_1\Pi\varphi+R_1\varphi,$$
where $R_1=\Pi[\Pi,P]\chi(P)=\mathcal O(h^{\infty})$. This equation can be re-written as,
$$ih\partial_t(e^{itP_1/h}\Pi\varphi)=\mathcal O(h^{\infty}\lVert \varphi_0\rVert).$$
and thus integrating from $0$ to $t$ we obtain,
$$\Pi\varphi=e^{-itP_1/h}\Pi\varphi_0+\mathcal O(\lvert t\rvert h^{\infty}\lVert\varphi_0\rVert).$$
In the same way we obtain
$$(1-\Pi)\varphi=e^{-itP_1/h}(1-\Pi)\varphi_0+\mathcal O(\lvert t\rvert h^{\infty}\lVert\varphi_0\rVert).$$
Combining these two equations we obtain \eqref{myeq2.11}.
\end{proof}

\begin{proof}[Proof of Theorem \ref{reduction}]
Since $\Pi-\Pi_0=\mathcal O(h)$, for $h$ small enough the operator
$$\mathcal U:=(\Pi_0\Pi+(1-\Pi_0)(1-\Pi)(1-(\Pi_0-\Pi)^2)^{-1/2}.$$
can be defined. $\mathcal U$ is a unitary operator and it maps $\mathrm{Ran}\Pi$ onto $\mathrm{Ran}\Pi_0$ (see Kato \cite{Ka} Chap.I.4). 

We can write $\Phi(P)\hat\Pi_j=\Phi(P)(P-i)^j(P-i)^{-j}\hat\Pi_j$. By the proof of Lemma \ref{bound2} $(P-i)^{-j}\hat\Pi_j$ is a h-admissible operator with a symbol in $ S^0(\mathbb R^3\times\mathbb R^3;\mathcal L(L^2_y,L^{2,s}_y))$ for any $s\in \mathbb R$. Writing $\Phi'(x):=\Phi(x)(x-i)^j$ we have
\begin{equation}\label{myeq4.92}
\Phi'(P)=\frac{i}{2\pi}\int_{\mathbb C}\bar\partial_z\tilde \chi'(z)(P-z)^{-1}dzd\bar z,
\end{equation}
where $\tilde \Phi'(z)$ is an almost analytic expantion (see e.g., Martinez \cite{Ma}). Expanding $(P-z)^{-1}$ in \eqref{myeq4.92} as in the proof of Lemma \ref{bound2} we can see that $\Phi'(P)$ is h-admissible, so that $\hat\Pi_{\Phi}$ is a h-admissible operator on $L^2(\mathbb R^3_x;L^2_y)$. By the equation
\begin{align*}
\Pi&=\oint_{\lvert z-1\rvert=1/2}(z-\hat\Pi_{\Phi})dz\\
&=\oint_{\lvert z-1\rvert=1/2}(z-\Pi_0)^{-1}\sum_{k=0}^{\infty}((\hat\Pi_{\Phi}-\Pi_0)(z-\Pi_0)^{-1})^kdz,
\end{align*}
we see that $\mathcal U$ is a h-admissible operator on $L^2(\mathbb R^3_x;L^2_y)$ and its princial symbol is $1$.

We define $W$ by
$$W\psi=(\mathcal U\psi,u_1(x,y))_{L^2_y(\mathbb R^{3N})}\oplus\dotsm\oplus(\mathcal U\psi,u_k(x,y))_{L^2_y(\mathbb R^{3N})}.$$
where $(u_1,\dots,u_k)$ is the orthnormal basis of $\mathrm{Ran}\Pi_0$ in (H3). Since $W^*(\alpha_1\oplus\dotsm\oplus\alpha_k)=\mathcal U^*(\alpha_1u_1+\dotsm+\alpha_ku_k)$ for $\alpha_i\in L^2(\mathbb R^3)$ we obtain
$$W^*W=\mathcal U^*\Pi_0\mathcal U=\Pi,\ WW^*=1,$$
which implies the unitarity of the restriction $U$ of $W$ to $\mathrm{Ran}\Pi$.

Difining $G:=UP_1U^*$ it is easy to see that $G$ is a h-admissible operator on $H^2(\mathbb R^3_x;L^2_y)$ and its symbol have the following form:
$$g(x,\xi)=\xi^2I_k+ \mu(x) +\sum_{i\geq 1} h^jg_j(x,\xi).$$

We shall prove $G$ is self-adjoint with the domain $U(\mathrm{Ran}\Pi\cap D(P_1))$. For this purpose we only need to prove $P_1$ is self-adjoint with the domain $\mathrm{Ran}\Pi\cap D(P_1)$ as an operator on $\mathrm{Ran}\Pi$ since $U$ is unitary on $\mathrm{Ran}\Pi$. Let $u,w\in\mathrm{Ran}\Pi$ and suppose $(u,P_1v)=(w,v)$ for all $v\in\mathrm{Ran}\Pi\cap D(P_1)$. Then for any $\tilde v\in D(P_1)$, $(u,P_1\Pi\tilde v)=(w,\Pi\tilde v)$. Moreover since $P_1=\Pi P\Pi$ and $w\in\mathrm{Ran}\Pi$, we have $(u,P_1(1-\Pi)\tilde v)=0$ and $(w,(1-\Pi)\tilde v)=0$. Thus we have $(u,P_1\tilde v)=(w,\tilde v)$ for any $\tilde v\in D(P_1)$. Since $P_1$ is self-adjoint, we see $u\in D(P_1)$ and therefore $u\in\mathrm{Ran}\Pi\cap D(P_1)$. Thus $P_1$ is self-adjoint with the domain $\mathrm{Ran}\Pi\cap D(P_1)$ as an operator on $\mathrm{Ran}\Pi$.

Since $G$ is self-adjoint we can define $e^{-itG}$ and obtain $U^*e^{-itG}U\Pi=e^{-itP_1}\Pi$. This establishes the theorem.
\end{proof}

\begin{proof}[Proof of Theorem \ref{coherent}]
In the same way as Martinez-Sordoni \cite[Theorem 11.3]{MaSo2}, for $\varphi_0=(\pi h)^{-3/4}\Pi\chi(P)(e^{ix\xi_0/h-(x-x_0)^2/2h}u_1(x))$ we obtain the following:
\begin{equation}\label{myeq4.93}
e^{-itP/h}\varphi_0=e^{i\delta_t/h}\sum_{\mu=0}^{3(J-1)}c_{\mu}(t;h)\phi_{\mu,t} v_{\mu}(x)+\mathcal O(h^{J/4}),
\end{equation}
where $\delta_t,c_{\mu}(t;h)$ and $\phi_{\mu,t}$ are satisfying the conditions in Theorem \ref{coherent} and $v_{\mu}(x)$ is satisfying the same condition as $\tilde v_{\mu}(x)$.

Since $\mathcal V^*\varphi_0=\tilde \varphi_0$, multiplying both sides of \eqref{myeq4.93} by $\mathcal V^*$ we have \eqref{myeq2.13} with $\tilde v_{\mu}(x)=\mathcal V^*v_{\mu}(x)$. Since $v_{\mu}(x)$ contains $u_1$, we have $$\partial_{x_j}(\mathcal V^*v_{\mu}(x))=-ie\left(A\sum_{i=2}^{N+1}y_i\right)_j\mathcal V^*v_{\mu}(x)+ \mathcal V^*\partial_{x_j}v_{\mu}(x)\in C^0(\mathbb R^3_x;L_y^2).$$
In the same way $\mathcal V^*v_{\mu}(x)$ is differentiable any number of times which completes the proof.
\end{proof}

\appendix
\section{Pseudodifferential operators with operator valued symbols}\label{fifthsec}
We introduce the classes of operator valued symbols we use.
\begin{dfn}
Let $\mathcal H_1$ and $\mathcal H_2$ be Hilbert spaces and $h_0$ a positive constant. A function $a(x,\xi;h)\in C^{\infty}(\mathbb R^n\times \mathbb R^{\ell}\times (0,h_0];\mathcal L(\mathcal H_1,\mathcal H_2))$ is said to be in $S^m(\mathbb R^n\times \mathbb R^{\ell};\mathcal L(\mathcal H_1,\mathcal H_2))$ if for any $\alpha \in \mathbb N^n$ and $\beta \in \mathbb N^m$ one has
\begin{equation}\label{S^m norm}
\sup_{(x,\xi,h)\in \mathbb R^{n+\ell}\times (0,h_0]}\lVert \partial_x^{\alpha} \partial _{\xi}^{\beta} a(x,\xi ;h)\rVert _{\mathcal L(\mathcal H_1, \mathcal H_2)}(1+\lvert \xi \rvert )^{-m+\lvert \beta \rvert}<\infty.
\end{equation}
\end{dfn}

\begin{dfn}
For $a(x,y,\xi)\in S^m(\mathbb R^n\times \mathbb R^n;\mathcal L(\mathcal H_1,\mathcal H_2))$ and $u\in\mathcal S(\mathbb R^n;\mathcal H_1)$ we set
$$Op_h^w(a)u(x;h) =\frac{1}{(2\pi h)^n}\int\left (\int e^{i(x-y)\xi /h}a(\frac{x+y}{2},\xi)u(y)dy\right )d\xi.$$
$Op_h^w(a)$ is called the semicllasical pseudodifferential operator with operator valued symbol.
\end{dfn}
We can see that $Op_h^w(a)$ is continuous $\mathcal S(\mathbb R^n;\mathcal H_1)\to \mathcal S(\mathbb R^n;\mathcal H_2)$ as in Martinez \cite{Ma}. Since the formal adojoint of $Op_h^w(a)$ is $Op_h^w(a^*)$, we can extend it uniquely to a linear continuous operator $\mathcal S'(\mathbb R^n;\mathcal H_1)\to \mathcal S'(\mathbb R^n;\mathcal H_2)$.

For the composition of pseudodifferential operators we have the following theorem as in the scalar case (see H\"ormander \cite{Ho}).
\begin{thm}\label{composition}
If $a_1\in S^{m_1}(\mathbb R^n\times \mathbb R^n;\mathcal L(\mathcal H_1,\mathcal H_2)),\ a_2\in S^{m_2}(\mathbb R^n\times \mathbb R^n;\mathcal L(\mathcal H_2,\mathcal H_3))$, then as operators on $\mathcal S(\mathbb R^n;\mathcal H_1)$ or $\mathcal S'(\mathbb R^n;\mathcal H_1)$
\begin{equation*}\label{compositioneq}
Op_h^w(a_2)Op_h^w(a_1) =Op_h^w(b),
\end{equation*}
where $b\in  S^{m_1+m_2}(\mathbb R^n\times\mathbb R^n;\mathcal L(\mathcal H_1,\mathcal H_3))$ is given by 
\begin{equation*}\label{compositioneq2}
b(x,\xi)=e^{ih[D_{\eta}D_x-D_yD_\xi]/2}a_2(y,\eta)a_1(x,\xi)\big \vert_{\substack{y=x\\ \eta=\xi}}=:a_1\# a_2,
\end{equation*}
and has the expansion
\begin{equation*}\label{asymptoticcomp}
b(x,\xi)= \sum_{\lvert\alpha+\beta\rvert\leq N}\frac{h^{\lvert \alpha+\beta\rvert}(-1)^{\lvert \alpha \rvert}}{(2i)^{\lvert \alpha+\beta\rvert}\alpha !\beta !}(\partial_x^{\alpha}\partial_\xi^{\beta}a_2(x,\xi))(\partial_\xi^{\alpha}\partial_x^{\beta}a_1(x,\xi))+h^{N+1}r_N(x,\xi),
\end{equation*}
where $r_N\in  S^{m_1+m_2-N-1}(\mathbb R^n\times\mathbb R^n;\mathcal L(\mathcal H_1,\mathcal H_3))$
\end{thm}

The $L^2$-boundedness can be established as in the scalar case.
\begin{thm}\label{L^2boundedness}
Let $a\in S^0(\mathbb R^n\times\mathbb R^n;\mathcal L(\mathcal H))$. Then $Op_h^w(a)$ is continuous on $L^2(\mathbb R^n;\mathcal H)$ and
$$\lVert Op_h^w(a)\rVert_{\mathcal L(L^2(\mathbb R^n;\mathcal H))}\leq C_n\left(\sum_{\lvert \alpha+\beta\rvert\leq M_n}\lVert \partial^{\alpha}_x\partial_{\xi}^{\beta}a\rVert_{L^{\infty}(\mathbb R^n\times \mathbb R^n;\mathcal L(\mathcal H))}\right),$$
where the positive constants $C_n$ and $M_n$ depend only on $n$.
\end{thm}

We also use the useful notion of h-admissible operators (See Martinez-Sordoni \cite{MaSo2} Appendix).

\begin{dfn}
Let $m\in \mathbb R$ and let $\mathcal H_1$ and $\mathcal H_2$ be Hilbert spaces. An operator $A=A(h):H^m(\mathbb R^n;\mathcal H_1) \to L^2(\mathbb R^n;\mathcal H_2)$ with $h\in (0,h_0]$ is called h-admissible (of degree m) if, for any $N\geq 1$
$$A(h)=\sum_{j=0}^Nh^jOp_h^w(a_j(x,\xi;h))+h^NR_N(h),$$
where $R_N$ is uniformly bounded from $H^m(\mathbb R^n;\mathcal H_1)$ to $L^2(\mathbb R^n;\mathcal H_2)$ for $h\in (0,h_0]$, and, for all $h>0$ small enough, $a_j\in S^m(\mathbb R^{2n};\mathcal L(\mathcal H_1,\mathcal H_2))$.
\end{dfn}

\subsection*{Acknowledgment}
The author would like to express his great appreciation to Professor Yoshio Tsutsumi for his helpful advices and encouragements. The author also shows his deep gratitude to Professor Tadayoshi Adachi for his helpful discussions.

\end{document}